\documentclass[letterpaper, 10 pt, conference]{ieeeconf}  

\IEEEoverridecommandlockouts                              
\usepackage{cite}\usepackage{hyperref}
\usepackage{amsmath,amssymb,amsfonts}
\usepackage{algorithmic}
\usepackage{graphicx}
\usepackage{textcomp}
\usepackage{amsmath,amssymb,bm,bbm,mathrsfs,amscd}
\usepackage{calc}
\usepackage{color}

\usepackage{dsfont}
\usepackage{graphicx}
\usepackage{epstopdf}
\usepackage{epsfig}
\usepackage{tikz}
\usepackage{amsmath}
\usepackage{amsfonts}
\usepackage{amssymb}
\usepackage{rotating}
\usepackage{mathtools}
\usepackage{color}
\usepackage{subfig}
\usepackage{enumerate,pgfplots}

\newtheorem{theorem}{Theorem}
\newtheorem{definition}{Definition}
\newtheorem{proposition}{Proposition}
\newtheorem{lemma}{Lemma}
\newtheorem{corollary}{Corollary}
\newtheorem{example}{Example}
\newtheorem{remark}{Remark}

\newcommand{\bma}{\left[\begin{matrix}}
	\newcommand{\ema}{\end{matrix}\right]}

\newcommand{\ba}{\begin{array}}
\newcommand{\ea}{\end{array}}

\newcommand{\be}{\begin{equation}}
\newcommand{\ee}{\end{equation}}

\newcommand{\mc}{\mathcal}

\def\1{\mathds{1}}
\def\0{\boldsymbol{0}}

\newcommand{\R}{\mathbb{R}}
\newcommand{\N}{\mathbb{N}}

\newcommand{\X}{\mathcal{X}}

\newcommand{\blue}{\color{black}}

\def\N{\mathbb{N}}

\def\R{\mathbb{R}}

\def\diag{{\rm diag}\,}
\tikzstyle{v_c}=[circle, draw,inner sep=2pt, minimum width=12pt, color=blue]
\tikzstyle{v_a}=[circle, draw,inner sep=2pt, minimum width=12pt, color=red]
\tikzstyle{edge} = [draw,thick,-,font=\small ]
\tikzstyle{label} = [draw,fill=black,font=\normalsize]
\def\I{{\mathcal I}}

\def\BibTeX{{\rm B\kern-.05em{\sc i\kern-.025em b}\kern-.08em
	T\kern-.1667em\lower.7ex\hbox{E}\kern-.125emX}}
\usepackage{tikz}
\usepackage{comment}

\newcommand{\Yinterp}{Y}
\newcommand{\Xinterp}{Z}
\newcommand{\xinterp}{z}
\newcommand{\yinterp}{y}

\newcommand{\comp}{\hat{\blockbound}}
\newcommand{\pcomp}{\hat{\blockbound}}

\newcommand{\bound}{G}
\newcommand{\blockbound}{\Gamma}
\newcommand{\boundcomp}{\widehat\bound}
\newcommand{\Dcomp}{\widehat D}
\newcommand{\Acomp}{\widehat A}
\newcommand{\ccomp}{\widehat c}
\newcommand{\Qcomp}{\widehat Q}

\newcommand{\energy}{P}

\newcommand{\noisev}{v}
\newcommand{\x}{x}
\newcommand{\xplus}{x_+}

\newcommand{\datax}{X} 
\newcommand{\dataxplus}{X_+} 
\newcommand{\datau}{U} 
\newcommand{\noisem}{V} 
\newcommand{\noiseplus}{\bma \noisev &  \noisem \ema} 

\newcommand{\normtwo}[1]{\sigma_\text{max}(#1)}

\newcommand{\Aset}{\mc L_L}
\newcommand{\Abound}{L^2 I_n}

\usepackage{cite}

\usepackage{calligra}
\usepackage{mathrsfs}
\usepackage{mathalpha}
\usepackage[normalem]{ulem}

\title{\LARGE \bf 
Interpolation Conditions for Data Consistency and Prediction in Noisy Linear Systems
	
}
\author{Martina Vanelli, Nima Monshizadeh, Julien M. Hendrickx
	\thanks{This work was supported by ARC via the SIDDARTA project. }
	\thanks{M. Vanelli and J. M. Hendrickx are with the ICTEAM Institute, UCLouvain, B-1348, Louvain-la-Neuve, Belgium	
	(email:	{\tt \footnotesize	\{martina.vanelli, julien.hendrickx\}@uclouvain.be}). N. Monshizadeh is with the Engineering and Technology Institute
	Groningen, University of Groningen, Nijenborgh 4, 9747 AG
	Groningen, The Netherlands. Email: 
		{\tt \footnotesize n.monshizadeh@rug.nl.}}
}

\begin{document}

	\maketitle
	\thispagestyle{empty}
	\pagestyle{empty}

	\begin{abstract}
    We develop an interpolation-based framework for noisy linear systems with unknown system matrix with bounded norm (implying bounded growth or non-increasing energy), and bounded process noise energy. The proposed approach characterizes all trajectories consistent with the measured data and these prior bounds in a purely data-driven manner. This characterization enables data-consistency verification, inference, and one-step-ahead prediction, which can be leveraged for safety verification and cost minimization.  Ultimately, this work represents a preliminary step toward exploiting interpolation conditions in data-driven control, offering a systematic way to characterize trajectories consistent with a dynamical system within a given class and enabling their use in control design. 	
\end{abstract}
	\section{Introduction}
Data-driven control has become a crucial aspect of modern control theory, offering powerful tools for system analysis and design \cite{markovsky2021behavioral}.  Many existing approaches rely on the assumption of persistence of excitation \cite{willems2005note}, which, in the absence of noise, ensures that the underlying dynamical system can be uniquely identified from the data. Under this assumption, data-driven methods often enable bypassing the explicit identification of the system matrix, although the data itself would still allow it \cite{de2019formulas, van2020willems, berberich2020data}. 
{\blue In practice,  data are often affected by noise, preventing unique identification of the system parameters. Recent research has addressed such challenges under bounded-noise assumptions, for example through the S-Lemma \cite{van2020noisy}, Petersen’s Lemma \cite{bisoffi2022data}, 
 and bounds on measurement errors \cite{li2026controller}. These approaches characterize the set of all models consistent with both the available data and the prescribed noise bounds, enabling 
robust guarantees for stability, safety \cite{ modares2023data} and  linear quadratic regulator (LQR) design \cite{zheng2025robust, zhao2025data}.
 However, the persistence of excitation condition is not always satisfied or even necessary \cite{van2020data}. Recent work has explored combining data with prior knowledge of a nominal model\cite{niknejad2023physics, berberich2022combining} or positive system \cite{miller2023data} to relax this assumption and improve robustness.}

We contribute to this line of research by proposing a general framework that exploits available structure and prior knowledge to extract the maximum possible information from available data, regardless of its quantity or quality. 
In this initial study, we consider noisy linear time-invariant (LTI) systems with a known input matrix $B$ (e.g., representing a known actuator configuration) and an unknown state-transition matrix $A$. We assume a bound on the system norm, which corresponds to a bound on energy increase. In particular, when this bound equals one, the system exhibits non-decreasing energy and thus is stable. Additionally, the process noise is assumed to have a bounded energy. {\blue Notably, we differ from prior work by not assuming any minimal amount of data nor the knowledge of a nominal model.} 

   \begin{figure}
	\centering
    \includegraphics[width=0.22\textwidth]{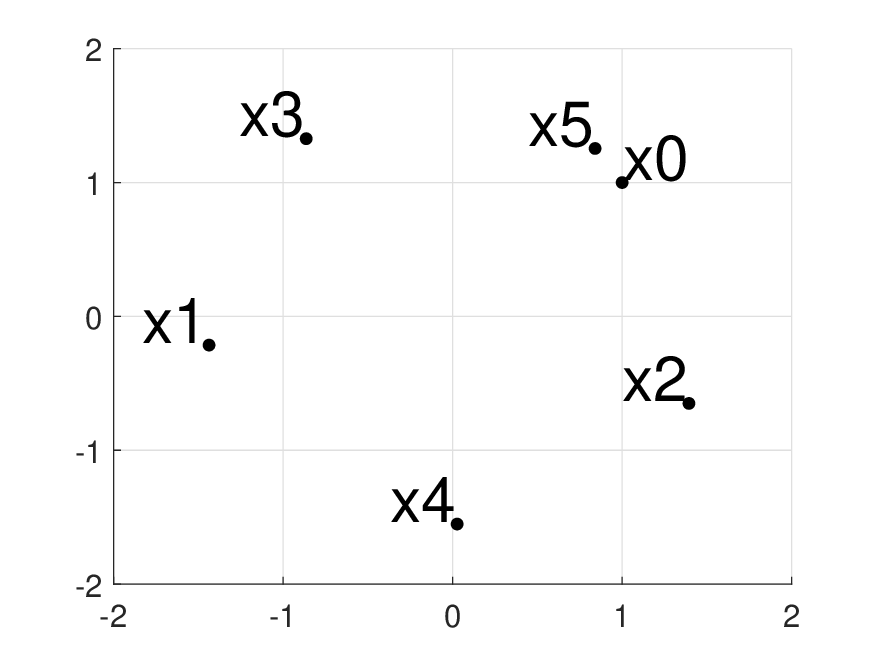}
\includegraphics[width=0.22\textwidth]{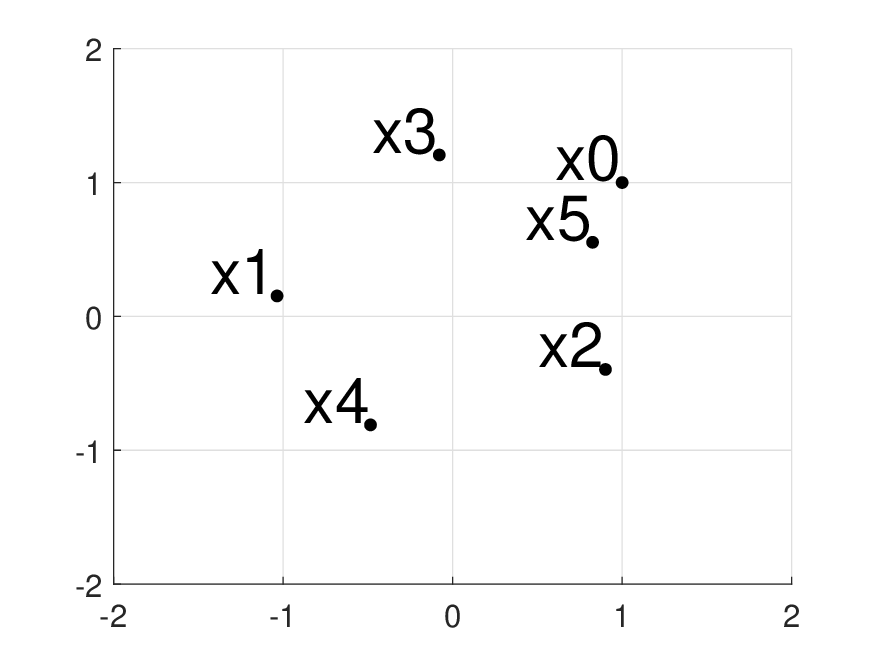}
	\caption{Collected data-points.}
   \label{fig:datapoints}
\end{figure}
 {\blue  As a pedagogical example, consider the data-points in Fig. \ref{fig:datapoints} collected from two systems: on the left, a noisy system that does not preserve energy; on the right, a system with strictly decreasing energy and a higher noise level.} 
Our goal is to address data-driven questions such as determining whether the observed data are compatible with given noise and energy bounds, identifying which parameter combinations explain the trajectories, and characterizing the set of feasible future states and optimizing over it.
{\blue Unlike the set membership estimation which seeks all models consistent with data and bounded uncertainty \cite{ lauricella2020set, li2024learning, xu2025sample},} 
our aim is to provide a systematic way to address such queries directly from data without explicitly identifying all admissible system matrices or noise realizations. 
This data-driven approach paves the way for extending the framework to the nonlinear setting, where explicit parametrizations of the dynamics are either unavailable or intractable.
 
To achieve this goal, we rely on
 \textit{interpolation conditions}, that is, necessary and sufficient conditions on a set of points that guarantee the existence of an interpolating function (or operator) belonging to a certain class (e.g. L-smooth convex functions \cite{taylor2017smooth}). The use of interpolation conditions has led to a novel type of analysis in optimization, enabling to derive exact worst-case performances of algorithms.
 This work is a preliminary step in the direction of exploiting interpolation conditions in data-driven control to characterize the set of trajectories that are consistent with a dynamical system in a given class.

Our main contributions are the following.
First, we derive suitable interpolation conditions for our setting. By means of a matrix-completion argument, we extend a classical algebraic result for bounded linear operators to the case of linear operators composed of multiple bounded components. 
%
Then, we exploit its application for LTI systems.
\subsubsection{Data consistency and inference} using the generalized interpolation conditions, we formulate semidefinite programs (SDP) that verify whether a dataset is compatible with assumed noise and  energy amplification (or decrease) bounds and that compute minimal admissible bounds.
\subsubsection{One-step-ahead prediction, safety, and energy minimization} We apply the algebraic result to characterize the set of all next states consistent with past data and prior assumptions. 
		This characterization can be used for safety verification   and worst-case (energy) minimization. 
		Since our formulation is not an SDP in general, we derive an explicit description of the feasible set: in the noise-free case it reduces to a single ellipsoid determined by the data and prior bounds; with process noise it becomes a \textit{union of ellipsoids}, whose parameters we explicitly derive.

 An important advantage of our preliminary results is that they do not require any assumptions on the amount of data. Indeed, our aim is to exploit the maximal information we can obtain from any amount of data and investigate how this information can be combined with the physical knowledge of the system to perform effective control.
The rest of the paper is organized as follows. 
We outline the problem setting in Section \ref{sec:1}, we present the interpolation conditions in Section \ref{sec:2} and  their applications to LTI systems in Section \ref{sec:3}. 
We conclude with Section  \ref{sec:4}, summarizing the results and outlining future research directions.

	\subsection{Notation}
For a matrix $M$ in $\R^{n\times m}$, we denote with $M^{\dagger}$ in $\R^{m \times n}$ its \textit{Moore-Penrose inverse} and we refer to it as the \textit{pseudo-inverse} of $M$. 
We denote with 
$\diag(M_1,\dots, M_n)$ a block diagonal matrix with diagonal blocks $M_1$, ..., $M_n$, with $I_n$ the identity matrix of size $n$, and with $0_{m\times n}$ the zero matrix of size $m\times n$. The largest and smallest singular values of a
 matrix $M$ are $\sigma_\text{max}(M)$ and $\sigma_\text{min}(M)$. The induced $2$-norm of a matrix $M$ is $\|M\|$ and is
 equivalent to $\sigma_\text{max}(M)$. We denote with
 $\mc L_{L}=\{M\mid \normtwo{M}\leq L\}$
the set of matrices whose norm is bounded by $L\geq 0$. 
For a symmetric matrix  $\tiny \bma M & N\\N^\top &O\ema$, we use the shorthand writing $   \tiny \bma M & N\\ \ast &O\ema$.
	\section{Problem setting}\label{sec:1}
We consider the discrete-time  linear time-invariant (LTI) system
	\be\label{eq:lin_sys}
x_{k+1}=Ax_k+Bu_k+v_k\,,
\ee
where $x_k\in \R^{n}$ is the state vector, $u_k \in \R^{m}$ is the input vector, and $v_k$ is the process noise. 
We assume that the input matrix $B$ is known, while the state-transition matrix $A$ is unknown. 
The assumption on the input matrix, that we intent to relax in future work, is satisfied when the actuator configuration is known, whereas the internal dynamics of the system is not.	From a methodological viewpoint, fixing 	$B$ allows us to isolate and analyze the uncertainty on $A$ and the impact of noise. 

We also assume to have access to  \textit{input-state measurements}, that is, $x_0, \dots, x_t$ and  $u_0, \dots, u_{t-1}$ for some $t>0$. Importantly, we make no assumption on the number or richness of the data and the number of samples $t$ may even be smaller than the system dimension $n$.
Let us define the following data matrices:
\be\label{eq:data}
\begin{aligned}
	\datax&=\bma x_{t-1}&\dots& x_0\ema\\
	\datau&=\bma u_{t-1}&\dots& u_0\ema \\
	\dataxplus&=\bma x_{t}& \dots& x_1\ema\,.
\end{aligned}
\ee
From \eqref{eq:lin_sys}, the data satisfy
\be\label{eq:matrix_form}
\dataxplus-B\datau=\bma A&\noisem\ema \bma \datax\\ I_t\ema\,,
\ee 
where 
$\noisem=\bma v_{t-1}&\dots& v_0\ema$
denotes the noise matrix.  Since $B$ and the data matrices $(\datax, \dataxplus,\datau )$ are known, 
the only unknowns in \eqref{eq:matrix_form} are $A$ and $\noisem$. 

\subsection{Prior knowledge on system dynamics} 
We assume prior knowledge that the system matrix norm is bounded, i.e.,
\be\label{eq:A_set}
	A\in 
\Aset := \{M : \sigma_\text{max}(M) \leq  L\}=\{M : M^\top M\preceq \Abound\}\,,
\ee
where 
$L>0$ is the \emph{energy amplification bound} of the system  in the absence of inputs or noise. Equivalently, we can express this condition in terms of a quadratic energy function. Assume there exists a  \emph{known} symmetric positive definite matrix $\energy=\energy^\top \succ 0$ such that
\begin{equation}\label{eq:energy}
	(A x)^\top \energy (A x) \le L^2x^\top \energy x
    \quad \forall x \in \mathbb{R}^n.
\end{equation}
Factorizing $P = R^\top R$ for some invertible matrix $R$ and defining the transformed coordinates $\tilde{x} = R x$ and $\tilde{A} = R A R^{-1}$, this condition becomes
\begin{equation}\label{eq:cv}	
(\tilde{A} \tilde{x})^\top (\tilde{A} \tilde{x}) \le L^2 \tilde{x}^\top \tilde{x}, \quad \forall \tilde{x} \in \mathbb{R}^n\,.
\end{equation}
or equivalently $\tilde A$ in $\Aset$. We remark that this transformation can be applied for a \emph{known} $P$. 

The value of $L$ characterizes three distinct regimes, that can be observed clearly in \eqref{eq:energy}. For $L<1$, the internal system exhibits strictly decreasing energy and the system is exponentially stable. In particular, the condition $\normtwo{A}< 1$ implies that the system is strictly contractive with contraction rate $\normtwo{A}$. 
For $L=1$, the system exhibits non-decreasing energy. In particular, \eqref{eq:energy} implies that the internal dynamics is  Lyapunov stable with respect to the Lyapunov function $\mc V(x) = x^\top P x$. For $L>1$, the system may amplify the state, but only in a controlled manner, which prevents unbounded instantaneous amplification of it. 

\subsection{Prior knowledge on the noise}
While the noise sequence is unknown, we assume that its \emph{total energy} admits a deterministic upper bound (see e.g. \cite{van2020data,bisoffi2022data, berberich2022combining}): 
\be \label{eq:bound_noise}\noisem\noisem^\top=\sum_{k=0}^{t-1}v_kv_k^\top \preceq \Theta \,.\ee 
 In particular, we focus on the case \be \label{eq:theta}\Theta =\alpha^2t  I_n\,. \ee 
This bound is automatically satisfied when all individual noise vectors are instantaneously bounded, i.e., $\|v_k\| \leq \alpha$, for all $k$. 
	 However, the converse is not generally true: the total energy bound does not exclude degenerate cases where, for example, all the noise energy is concentrated in a single time step  (e.g., $\|v_j\| = \alpha\sqrt{t}$  for some $j$ and $\|v_i\| = 0$ for $i \neq j$).  
Thanks to the structure of the noise bound \eqref{eq:theta}, we find that \eqref{eq:bound_noise} corresponds to 
$\|V\|= \sigma_\text{max}(\noisem^\top )=\sigma_\text{max}(\noisem)\leq \alpha \sqrt{t}$.  
Thus, the admissible noise matrices belong to the set
\be\label{eq:V_set}
\mc L_{\alpha\sqrt{t}} := \{M: \sigma_\text{max}(M) \leq  \alpha \sqrt{t}\}\,.
\ee
{ \begin{remark}\label{rem:lyap}
The current noise model introduces some conservatism when applying a change of variables. Specifically, given the energy inequality in \eqref{eq:energy}
	for 
	$P = P^\top \succ 0$ with $P = R^\top R$ and invertible $R$, applying the transformations
    $$
	\tilde{x} = R x, 
	\quad 
	\tilde{A} = R A R^{-1},\quad
	\tilde{B} = R B, 
	\quad 
	\tilde{v}_k = R v_k, 
	$$
	gives $\tilde{A} \in \mathcal{L}_L$ (see \eqref{eq:cv}) and $\tilde{V} \tilde{V}^\top \preceq \alpha^2 t P$. This tight relation can be replaced with the relaxation 
	\begin{equation}
		\tilde{V}^\top \tilde{V} 
		= (R V)^\top (R V) 
		= V^\top P V\preceq \lambda_{\max}(P)\, \alpha^2 t I_n\,,
	\end{equation}
	leading to $\tilde{V} \in \mathcal{L}_{\alpha \sqrt{t \lambda_{\max}(P)}}$. 
    Future work aims to develop interpolation frameworks that can handle more bounds (e.g., \eqref{eq:bound_noise} for a general $\Theta$), preventing this conservatism. 
\end{remark}}
\subsection{Objectives}\label{ss:goals}
Given the above setup, our goal is to address two main classes of problems. 
\subsubsection{Data consistency and inference}
Given the bounds $\alpha\geq 0$ and $L>0$, and the data $\datax,\dataxplus, \datau$, we aim to determine necessary and sufficient conditions on the data points for the existence of $A$ in $\mc L_L$ and $V$ in $\mc L_{\alpha \sqrt{t}}$ such that \eqref{eq:matrix_form} is satisfied. 
	
This characterization allows us to assess data consistency, heterogeneity or corruption. It also leads to inference problems that quantify the uncertainty or energy amplification required to explain the data.
The \emph{minimal noise level} \be\label{eq:mina}\min_{\alpha, A, V} \, \alpha \ \text{s.t.} \ A \in \mathcal{L}_1, \, V \in \mathcal{L}_{\alpha \sqrt{t}}, \, \eqref{eq:matrix_form}\ee 
quantifies the smallest disturbance level that makes the data consistent with a stable energy-minimizing model (possibly for a given Lyapunov function $P$, see Remark \ref{rem:lyap}). 
A large minimal noise suggests model mismatch, corrupted data, or unmodelled dynamics. 
 The \emph{minimal energy amplification bound} 
\be\label{eq:minl}\min_{L, A, V} \, L \ \text{s.t.} \ A \in \mathcal{L}_L, \, V \in \mathcal{L}_{\alpha \sqrt{t}}, \, \eqref{eq:matrix_form}\ee
gives the smallest amplification bound needed to explain the data for a given noise level. Values $L\leq 1$ imply stability, while $L>1$ quantifies the system’s maximal energy amplification.
Finally, both uncertainty and energy amplification can be jointly optimized,
providing a trade-off between noise and system energy. 

Note that the quadratic constraints on $A$ and $V$ in \eqref{eq:mina} and \eqref{eq:minl} 
can be rewritten as linear matrix inequalities (via the Schur complement), yielding an SDP. However, our objective is to address them in a \emph{data-driven} manner, without explicitly computing the matrices $A$ and $V$, thereby paving 
the way toward set-based predictions and the nonlinear setting.

\subsubsection{One-step ahead prediction} Given the bounds $\alpha\geq 0$ and $L>0$, the data $\datax,\dataxplus, \datau$ and the initial condition $x=x_t$,  we aim to characterize the set of next states that are consistent with the past measurements as a function of the input $u$. More precisely, 
we aim to give an exact characterization of the set
\be\label{eq:x_feas}
\begin{aligned}
\mc X_\text{+}(u):=\{& Ax+Bu+ v\,,\\&\quad 
\forall  A \in \mc L_L\,,\; \bma v& V\ema \in \mc L_{\alpha \sqrt{t+1}} \text{ s.t. } \eqref{eq:matrix_form}\}\,.
\end{aligned}
\ee
Again, note that our characterization is purely data-based and does not require fixing or identifying any specific model.
The set \eqref{eq:x_feas} can be used for \textit{safety verification} and \textit{optimal control design}. Examples of such problems are:
\begin{itemize}
	\item \emph{Safety:} Given a target region $\mathcal{B} \subset \mathbb{R}^n$, find a control $u$ such that all consistent next states remain in $\mathcal{B}$:
	\[
	\text{find } u \text{ s.t. } \mathcal{X}_+(u) \subseteq \mathcal{B}.
	\]
	
	\item \emph{Worst-case control / energy minimization:} Minimize the worst-case cost over all next states:
	\[
	\min_{u \in \mathbb{R}^m} \max_{x \in \mathcal{X}^+(u)} x^\top Q x + u^\top R u,
	\]
	with $Q \succ 0, R \succ 0$. This corresponds to \textit{one-step LQ control}.
\end{itemize}
While the focus of this paper is on one-step-ahead predictions, the extension to multiple steps ahead is fully compatible with our approach. However, this introduces additional computational complexity, which will be explored and discussed in future work. 

\section{Interpolation conditions}\label{sec:2}
The aim of this section is to determine necessary and sufficient conditions under which two given matrices,  $\Xinterp$ and $\Yinterp$ (representing the observed data), can be interpolated by a linear operator $M$ (i.e., $\Yinterp=M\Xinterp$) that satisfies prior structural or physical constraints such as bounded energy constraints. 
As a simple motivating example, consider the autonomous system 
\be\label{eq:aut_sys}
x_{k+1}=Ax_k
\ee
where $A$ in $\R^{n\times n}$ is unknown. Given the measurements $x_0, \dots, x_t$, 
we aim to determine if there exists a matrix $A$ within a given constraint set (e.g., $A^\top A\preceq L^2I_n$) such that $\dataxplus=A\datax$ with $\datax$ and $\dataxplus$ defined as in \eqref{eq:data}. 

This question can be addressed using an existing algebraic result, which
appears in various forms, including Proposition 1 in \cite{bisoffi2024controller}. It can be derived from the Douglas Lemma \cite{douglas1966majorization} and the special case $G=\Abound$ also follows from Theorem 3.1 in \cite{bousselmi2024interpolation}.

\begin{proposition}[{\cite[Prop.~1]{bisoffi2024controller}}]
\label{pr:1}
	Given the \emph{data} $\Xinterp$ in $\R^{m\times t}$ and $\Yinterp$ in $\R^{n\times t}$ and the \emph{bound} $\bound\in \R^{m \times m}$, $\bound\succeq 0$, there exists 	$M \in \R^{n\times m}$ such that
	\be\label{eq:bound_gen}
	\Yinterp = M\Xinterp\,,\qquad M^{\top}M\preceq \bound
	\ee 
	if and only if
	\be 
	\Yinterp^{\top} \Yinterp\preceq \Xinterp^{\top} \bound\Xinterp\,.	\ee 
\end{proposition}
\medskip
Proposition \ref{pr:1} provides a straightforward feasibility test: in the previous example, there exists a matrix
$A$ satisfying $A^\top A\preceq I_n$ and $\dataxplus=A\datax$
if and only if
$$
\dataxplus^\top \dataxplus\preceq L^2\datax^\top \datax\,.
$$

Using the same framework, we can characterize future states consistent with the data and the prior knowledge, i.e.,
\begin{equation}
	\mathcal{Y}_+ := \{\, y \in \mathbb{R}^{n} : y = Mz,\; M^\top M \preceq G,\; Y = MZ \,\}\,,\end{equation}
    where $\xinterp$ is the current state. The additional condition $\yinterp =M\xinterp$ 
can be encoded  in the framework of Proposition \ref{pr:1} by imposing
$$\bma \yinterp& \Yinterp\ema =M\bma \xinterp& \Xinterp\ema\,,$$ 
thus obtaining the necessary and sufficient condition 
$$
\bma \yinterp& \Yinterp\ema^\top \bma \yinterp& \Yinterp\ema \preceq \bma\xinterp& \Xinterp\ema^\top G \bma\xinterp& \Xinterp\ema\,.
$$
Interestingly, when the data are consistent with the bound, the set $\mathcal{Y}_+$ forms an ellipsoid, as stated in the following proposition (proved in Appendix  \ref{app1}). 
\begin{proposition}\label{pr:no_noise}
	Given the \emph{data} $\Xinterp$ in $\R^{m\times t}$ and $\Yinterp$ in $\R^{n\times t}$ and the \emph{bound} $\bound\in \R^{m\times m}$,  $\bound\succeq 0$, let
	\be\label{eq:D}
	D:=\Xinterp^{\top} \bound\Xinterp-\Yinterp^{\top} \Yinterp
	\ee
	and assume that $\text{rank}(D)=\text{rank}(\Xinterp)$. Then:
	\begin{enumerate}
		\item If $D\nsucceq 0$, $\mc Y_+=\emptyset\,$ (no feasible future states exist).
		\item  If $D\succeq 0$, for every \emph{current state} $\xinterp$ in $\R^m$, 
		\be\label{eq:ell} 
		\mc Y_+=\{\yinterp \in \R^n\,\mid\,(\yinterp-c)^{\top}\mc A(\yinterp-c) \leq Q\}
		\ee
		where $\mc A =  I_{n}+\Yinterp D^{\dagger}\Yinterp^{\top}\succ 0$ and $c$ and $Q\geq0$ can be computed from the data and the bound (see  \eqref{eq:par_ell} in Appendix \ref{app1}).
		\end{enumerate}
\end{proposition}
\medskip
Proposition \ref{pr:no_noise}  provides  a complete characterization for the set of feasible next states consistent with the data and the system's bound when $\text{rank}(D)=\text{rank}(\Xinterp)$.
This rank condition is primarily a technical requirement. 
{\blue For instance, for $G=I_n$, $\Xinterp=\bma0&1 \ema^\top $ and $\Yinterp =\bma1&0 \ema^\top$, we find that $D=0$ and thus $\text{rank}(D)=0<\text{rank}(\Xinterp)$. 
If we slightly increase $L$ by some $\epsilon>0$, we obtain $D_\epsilon>0$ and Proposition \ref{pr:no_noise} applies. 
Intuitively, the ellipsoid collapses into 
an ellipsoid in a lower-dimensional space when the bound perfectly aligns with the data. }
For simplicity, we omit this case here, but it will be addressed in the full version of the paper.

Proposition \ref{pr:1} 
apply when prior knowledge about $M$ is captured by a single energy bound $M^\top M \preceq \bound\,.$
To extend this to the setting described in Section \ref{sec:2}, we need to generalize the result to accommodate multiple bounded components, such as separate bounds on $A$ and $V$. 
We therefore seek conditions for the existence of a block-structured linear operator
\be\label{eq:M}
M=\bma M_1& \dots &M_N\ema \,,
\ee
where each block is subject to its own bound
\be\label{eq:multi_b}
M_i^\top M_i\preceq \blockbound_i
\ee
for some 
symmetric positive definite matrices $\blockbound_i\succeq 0$. 
As we shall see, this problem can be reformulated as a positive semidefinite matrix completion problem.
\begin{definition}
	Let
	\[
	G_{\Gamma} =
	\bma
	\Gamma_1 & ? & \cdots & ? \\[4pt]
	? & \Gamma_2 & \cdots & ? \\[2pt]
	\vdots & \vdots & \ddots & \vdots \\[2pt]
	? & ? & \cdots & \Gamma_N
	\ema
	\]
	be a partially specified block matrix, whose diagonal blocks 
	$\Gamma_i \in \mathbb{R}^{m_i \times m_i}$ are known, 
	while off-diagonal blocks (marked by ``$?$'') are unspecified.  
	A \emph{positive semidefinite (PSD)  matrix completion} of $G_{\Gamma}$ is any symmetric PSD matrix 
	$$
	\boundcomp = (\boundcomp_{ij})_{i,j=1}^N
	$$
	such that each diagonal block satisfies 
	$\boundcomp_{ii} = \Gamma_i$.
\end{definition}
In words, a PSD completion fills the unspecified blocks of  $\bound_\blockbound$ while preserving overall positive semidefiniteness.

The next result generalizes Proposition \ref{pr:1} to the case of multiple energy bounds.
\medskip
\begin{proposition}[Multi-bound extension]\label{pr:multi_b_int}
	Given the data $\Yinterp\in \R^{n\times t}$ and $\Xinterp_i\in \R^{m_i\times t}$ for $i=1,\dots, N$, and symmetric bounds $\blockbound_i \in \R^{m_i\times m_i}$, $\blockbound_i\succeq 0$, there exist matrices $M_i\in \R^{n\times m_i}$ satisfying
	\begin{equation}\label{eq:m_mat}
		\Yinterp = \sum_{i=1}^N M_i \Xinterp_i,
		\qquad
			M_i^{\top}M_i \preceq \blockbound_i\,, \quad i=1,\dots, N\,,
	\end{equation}
	if and only if there exists a PSD completion $\boundcomp$ of $\bound_\blockbound$ such that
	\begin{equation}\label{eq:bound_m_matr}
		\Yinterp^{\top}\Yinterp \preceq \Xinterp^{\top}\boundcomp\Xinterp,
		\quad \text{with}\quad
		\Xinterp = \bma \Xinterp_1 \\ \vdots \\ \Xinterp_N \ema.
	\end{equation}
\end{proposition}
\medskip
		\begin{proof}
	\emph{($\Rightarrow$)}   Assume that matrices $M_i$ exist satisfying \eqref{eq:m_mat}. Then, define
	$$
	\boundcomp =
	\bma
	\blockbound_1 & M_1^{\top}M_2 & \cdots & M_1^{\top}M_N \\[4pt]
	\ast & \blockbound_2 & \cdots & M_2^{\top}M_N \\[2pt]
	\vdots & \vdots & \ddots & \vdots \\[2pt]
	\ast & \ast & \cdots & \blockbound_N
	\ema.
	$$
Let $M$ be given by \eqref{eq:M}.
	By construction,
	$$
		\begin{aligned}
		 \boundcomp-M^\top M&=\diag(\Gamma_1-M_1^\top M_1,\dots, \Gamma_N-M_N^\top M_N)\\&\succeq 0\,,
			\end{aligned}
	$$
	since each diagonal block is PSD by assumption, while the off-diagonal blocks are zero.
	This yields $\boundcomp \succeq M^\top M \succeq 0$.  
Therefore,  $\boundcomp$ is a PSD completion of $\bound_\blockbound$ and satisfies both conditions needed to apply Proposition~\ref{pr:1}: $\boundcomp \succeq 0$ and $ M^\top M \preceq \boundcomp$. By applying the proposition with $G=\boundcomp$, $M$ given by \eqref{eq:M} and $Z$ in \eqref{eq:bound_m_matr}, we obtain the inequality in \eqref{eq:bound_m_matr}. 
    
	\emph{($\Leftarrow$)}
By the if-direction of Proposition~\ref{pr:1} with $\bound=\boundcomp$, there exists $M$ such that $\Yinterp=M\Xinterp$ and $M^\top M\preceq \boundcomp$.	
Since the block diagonal matrix $\boundcomp-M^\top M$ is PSD, its 
diagonal blocks are also PSD, which yields the  second inequality in \eqref{eq:m_mat}. 
\end{proof}

The direct analogous to Proposition \ref{pr:no_noise} requires now a union of ellipsoid (one for each completion), as will be discussed in Section \ref{ss:ell}.

\section{Data-consistency and prediction in LTI systems}\label{sec:3}
We now exploit the obtained results (Proposition \ref{pr:no_noise} and \ref{pr:multi_b_int}) to address the objectives outlined in Section \ref{ss:goals}.
\subsection{Data-consistency and inference}\label{ss:ver}
Recall that the data satisfy \eqref{eq:matrix_form}, that is, 
$$
\dataxplus-BU=\bma A&\noisem\ema \bma \datax\\ I_t\ema
$$
where $\datax$, $\datau$ and $\dataxplus$ are defined in \eqref{eq:data} and $\noisem$ is the noise. 
This expression fits the general framework of Proposition \ref{pr:multi_b_int} by setting
\be
M=\bma A&\noisem\ema\,,\quad	\Xinterp:=\bma \datax\\I_n\ema\,,\quad  
	\Yinterp:=	\dataxplus-BU\,.
\ee
The corresponding energy bounds 
on $A$ and $V$ can be encoded into
$\Gamma=\{\Abound\,, \alpha^2tI_t\}$. We then obtain the following.

\begin{corollary}[Data-consistency]\label{cor1} 
There exist $A$ in $\Aset$ and $V$ in $\mc L_{\sqrt{t}\alpha}$ such that \eqref{eq:matrix_form} holds true
if and only if there exists a matrix $\comp$ in $\R^{n\times t}$ satisfying
		\be\label{eq:ver1}
			\bma \Abound&\comp\\ \ast & \alpha^2tI_t\ema \succeq 0\,,\\
		\ee and
		\be\label{eq:ver2}
		\begin{aligned}
			L^2\datax^{\top} \datax&+\datax^{\top}\comp+\comp^{\top}\datax+t\alpha^2  I_t\\&-(\dataxplus-B\datau)^\top (\dataxplus-B\datau)\succeq0\,.
		\end{aligned}
		\ee
\end{corollary}
\medskip

 Both \eqref{eq:ver1} and \eqref{eq:ver2} are linear in $\comp$, allowing feasibility to be checked via SDP. 
		This yields a computationally efficient, data-driven method to assess whether the observed data are consistent with the assumed noise and system properties.  
        Furthermore, it provides a foundation (see Section \ref{ss:pred}) for set-based prediction methods as well as potential nonlinear extensions. The optimization problems outlined in Section \ref{ss:goals}, that is, \eqref{eq:mina},  \eqref{eq:minl} and  
        combinations of them, can likewise be reformulated by optimizing over 
		$\comp$	and replacing the constraints with \eqref{eq:ver1} and \eqref{eq:ver2}.
{\blue \begin{example}\label{ex:ver}
Consider the six data-points plotted on the left of Fig. \ref{fig:datapoints}, that is $\tiny X=\bma
0.03 & -0.86 & 1.39 & -1.44 & 1\\
-1.55 & 1.33 & -0.65 & -0.21 & 1
\ema$ and $\tiny \dataxplus= \bma 0.84 & 0.03 & -0.86 & 1.39 & -1.44\\
1.25 & -1.55 & 1.33 & -0.65 & -0.21\ema$,  
and let $U=0_{m\times 1}$.
By Corollary \ref{cor1}, there exist $A$ in $\mc L_L$ and $V$ in $\mc L_{\alpha\sqrt{t}}$ interpolating the data if and only if 
 there exists $F$ in $\R^{2\times 6}$ satisfying the conditions in
\eqref{eq:ver1} and \eqref{eq:ver2}. Let $L=1$. Then, by minimizing the noise level over $\alpha$ and $F$ such that \eqref{eq:ver1} and \eqref{eq:ver2} is satisfied, we find    $\alpha^*\approx 0.06$.  If we now set $\alpha=0.06$ and minimize the energy amplification bound,  
we find $L^*= 1$. Finally, if we minimize the noise level $\alpha$ over $L$, $\alpha$ and $F$, 
we get $\alpha^*=0.01$ for $L^*=2.9$.  When both parameters are optimized jointly, the noise level decreases while the energy bound increases, showing that allowing higher energy enables lower noise.  Consider now the data on the right of Fig. \ref{fig:datapoints}, that is,
$\tiny X=\bma
-0.5 & -0.08 & 0.9 & -1.04 & 1\\
-0.8 & -1.2 & -0.4 & 0.15 & 1
\ema$ and 
$\tiny \dataxplus=\bma0.83 & -0.5 & -0.08 & 0.9 & -1.04\\0.55 & -0.8 & -1.2 & -0.4 & 0.15\ema$. 
Then, the minimum noise level for $L=1$ is $\alpha^*\approx 0.32$, while the minimum energy amplification bound for $\alpha=0.32$ is $L^*\approx 0.83$. Finally, the minimum noise level obtained by letting both $\alpha$ and $L$ vary is $\alpha^*=0.32$ (for $L^*=6.18$). Notably, the noise level that we first obtain is already minimal and, when minimizing the energy amplification bound for such $\alpha$,  we obtain a lower $L$. 
Since the noise cannot be further reduced, joint optimization leaves $\alpha$ unchanged, confirming it has reached its minimum. 
These results suggest that the first set of data-points has lower noise compared to the second one, while the second one has a higher decrease in energy, which is consistent with the true value of the system matrix and the noise. We remark that these results are obtained in a data-driven manner, by constructing a feasible completion $F$ instead of computing candidates for $A$ and $V$. 

\end{example}
}
\subsection{One-step ahead prediction}\label{ss:pred}
We now apply Proposition \ref{pr:multi_b_int} to characterize the set of all next states $\xplus$ consistent with past observations and prior bounds, as a function of the input $u$, namely, $\mc X_+(u)$ as defined in \eqref{eq:x_feas}.

\begin{theorem}[One-step ahead prediction]\label{th1} 
Consider the data $\datax$, $\dataxplus$ and $\datau$, the bounds $\alpha\geq 0$ and $L>0$ and the initial condition $\x=x_t$. Let $\bar{X}=\bma \x&\datax\ema $. Then, we find 
\be\mc X_+(u)=\{\xplus \in \R^n \mid  \exists \comp \text{ s.t. } \eqref{eq:pred2}\,, \eqref{eq:pred3}\}\ee
where
\be \label{eq:pred2}
\bma \Abound&\pcomp\\ \ast & \alpha^2(t+1)I_{t+1}\ema \succeq 0\,,\\
\ee and 
\be\label{eq:pred3}
\begin{aligned}
	&\bma (\xplus-Bu)^{\top}(\xplus-Bu)&(\xplus-Bu)^{\top}(\dataxplus-B\datau)\\ \ast &(\dataxplus-B\datau)^{\top}(\dataxplus-B\datau)\ema\\[2pt]&\qquad\qquad\preceq	\bar{X}^{\top} \bar{X}+\bar{X}^{\top}\pcomp+\pcomp^{\top}\bar{X}+(t+1)\alpha^2  I_{t+1}\,.
\end{aligned} 		\ee
\end{theorem}
\medskip
\begin{proof}
We aim to characterize $\mc X_+(u)$ as defined in \eqref{eq:x_feas} in a data-based manner. Let  $\xplus=A\x+Bu+\noisev$. Then, this condition can be combined with \eqref{eq:matrix_form}, by writing the compact form 
	\be\label{eq:matrix_form2}
	\bma \xplus-Bu&\dataxplus-B\datau\ema=\bma A&\bma \noisev&\noisem\ema  \ema \bma  \x&\datax \\   1&0_{1\times t}\\ 0_t &I_t\ema 
	\ee
Let $z$ and $Z$ be 
\be\label{eq:X_with_noise}
	\xinterp = \bma
	\x \\[2pt]
	1 \\[2pt]
	0_{t}
	\ema \,, \quad \Xinterp= \bma 
	\datax \\[2pt]
	0_{1\times t} \\[2pt]
	I_t
	\ema \,.
	\ee 
Applying Proposition \ref{pr:multi_b_int} with  
$$\begin{aligned}\bar{M}=&\bma A&\noiseplus \ema\,,\quad \bar{Z}=\bma \xinterp &\Xinterp\ema\,, \\
	&\bar{Y}=\bma \xplus-Bu&\dataxplus-B\datau\ema\,,\quad  
	\end{aligned}$$  and $\Gamma=\{\Abound,(t+1)\alpha^2I_t \}$ yields the result. 
\end{proof}
\medskip
Theorem \ref{th1} gives an \textit{exact} characterization of the set of feasible points. 
For any given $u$, verifying feasibility for a candidate 
$\xplus$ reduces to an SDP feasibility check, since the conditions are linear in $\pcomp$ (analogous to the verification problem in Section \ref{ss:ver}). On the other hand, optimizing over the entire set 
$\mc X_+(u)$, e.g.,  computing, for a given $u$, the worst-case scenario
$$
\begin{aligned}
\max_{\xplus\in \X_+(u)} &\xplus^\top \mc Q \xplus + u^\top Ru =\\ \max_{\xplus, \comp}\quad  &\xplus^\top Q \xplus + u^\top Ru\; \text{ s.t. }\; \eqref{eq:pred2}\,,\,\eqref{eq:pred3}
\end{aligned}
$$ is not an SDP problem, as the inequality in \eqref{eq:pred3} depends quadratically on $\xplus$. To gain more insights, we next show that this feasible set can be represented as a union of ellipsoids, each corresponding to a valid PSD completion. 
\subsection{Ellipsoidal feasible sets}\label{ss:ell}
 To gain intuition on the structure of feasible predictions, we first focus on the noise-free case. This setting allows us to apply directly Proposition \ref{pr:no_noise} (with $\Yinterp = \dataxplus -B\datau$, $\Xinterp =\datax$ and $G=\Abound$) and characterize the set of feasible next states as an ellipsoid. 
\begin{corollary}\label{cor:no_noise}
Consider the data $\datax$, $\dataxplus$ and $\datau$, the bound $L>0$  and the initial condition $\x=x_t$. Let $\alpha=0$ and
	\be\label{eq:D_data}
	D:=L^2\datax^{\top} \datax-(\dataxplus-B\datau)^{\top} (\dataxplus-B\datau)
	\ee
	and assume that $\text{rank}(D)=\text{rank}(\datax)$. Then:
	\begin{enumerate}
		\item If $D\nsucceq 0$, $\mc X_+(u)=\emptyset\,$ (no feasible future states exist).
		\item  If $D\succeq 0$, for every \emph{initial condition} $\xinterp$ in $\R^m$, 
		\be\label{eq:ell_no_noise} 
		\hskip -0.6cm \mc  X_+(u)=\{\xplus \in \R^n\,\mid\,(\xplus-Bu-c)^{\top}\mc A(\xplus-Bu-c) \leq Q\}
		\ee
		where
		\be\label{eq:par_ell_no_noise}
		\begin{aligned}	\mc A &=  I_{n}+(\dataxplus-B\datau) D^{\dagger}(\dataxplus-B\datau)^{\top}\succ 0\\
			c & =(\dataxplus-B\datau)\datax^\dagger \x\\
			Q &=L^2x^{\top}(I_n-\datax D^\dagger D\datax^\dagger)\x\,.
		\end{aligned}
		\ee 
	\end{enumerate} 
\end{corollary}
\medskip

When the data are consistent with the physical assumption on the norm (i.e., $D\succeq 0$), the set of feasible states forms an ellipsoid with the parameters in \eqref{eq:par_ell_no_noise}.  
	Observe that $c=A^*x$ where $A^*=(\dataxplus-B\datau) \datax ^\dagger$ is the solution of the \textit{least-square} problem for $\dataxplus-B\datau =A\datax$. 
	All the possible matrices consistent with the data are given by $A=A^*+\Delta(I_m-\datax \datax^\dagger)$ for every $\Delta\in \R^{n\times m}$. 
	The further knowledge on the bound ($A^\top A\preceq L^2I_n$) allows us to reduce the set of feasible states at the next step to an ellipsoid. 
			\begin{remark}	
It is illustrative to look at $\mc X_+(u)$ in two extreme cases. First, if no data is available, i.e., $\datax=0_n$, $\dataxplus=0_n$ and $\datau=0_m$, we find the set of feasible states 
	$$
	\mc X_+(u)=\{ \xplus \in \R^n \mid (\xplus-Bu)^\top (\xplus-Bu) \leq L^2\x^\top \x\}\,,
	$$
	i.e.,	a ball centered in $Bu$ of radius $L\|x\|$. 
	On the other hand, if 
$\text{rank}(D)=\text{rank}(\datax)=n$, we find $c=(\dataxplus-BU) \datax ^{-1}\x$ and $Q=0$, which is consistent with uniquely identifying the state matrix $A^*=(\dataxplus-BU) \datax ^{-1}$. 
\end{remark}
 \begin{figure}
	\centering
\includegraphics[width=0.22\textwidth]{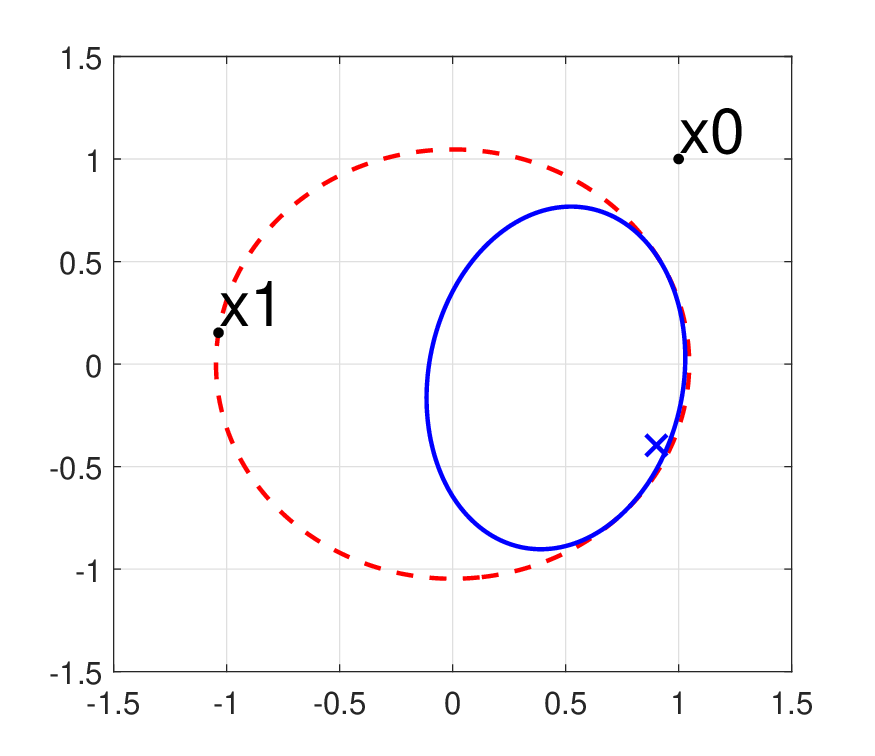}
\includegraphics[width=0.22\textwidth]{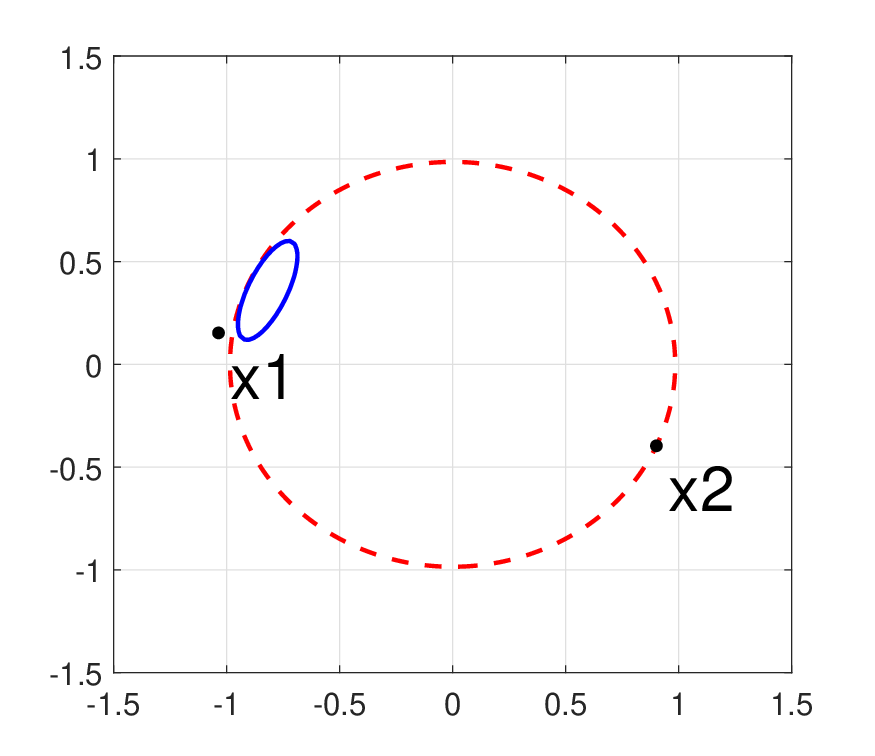}
	\caption{Representation of Example \ref{ex:ell_no_noise}}
    \label{fig:ell}
\end{figure}
\begin{example}\label{ex:ell_no_noise}
Let $n=2$ and assume that we have collected two data points $x_0$ and $\x_1$. For $L=1$, we find $D=\|x_0\|^2-\|x_1\|^2\,.$ Then, if $D<0$, the data are not consistent with the unit bound on $A$, leading to no feasible future states. On the other hand, if $D>0$, we find the ellipsoid in \eqref{eq:ell_no_noise} with $\mc A=I_n-  \frac{x_1x_1^\top}{\|x_0\|^2-\|x_1\|^2}$, $  c=\frac{x_0^\top x_1}{\|x_0\|^2}x_1$ and $Q=\frac{(x_0^\top x_1)^2}{\|x_0\|^2}$. 
Finally, if $D=0\,,$
we have that $\text{rank}(D)=0<\text{rank}(\Xinterp)=1$ and Corollary \ref{cor:no_noise} does not apply. In this case, the set of feasible points is given by 
a segment (an ellipse in dimension 1). 
Consider now the states on the left-hand side of Figure \ref{fig:ell}, that is, $x_0=\bma 1&1\ema^\top$ and $x_1= \bma-1.04
&0.15\ema^\top$. The area inside the dashed red circle is the feasible set for $x_2$ knowing only the current state $x_1$ and the bound $L=1$. When also knowing $x_0$, the area reduces to the internal part of the ellipse shown in blue. The additional information of $x_0$ significantly reduces the area for the next point and changes the shape of the ellipse. 
Assume further that we observe the realization of $x_2=
\bma0.9&-0.4\ema^\top $ (depicted with a blue cross on the left). On the right hand side of Figure \ref{fig:ell}, 
we show the feasible sets for $x_3$  when only knowing $x_2$ (dotted red area) and when knowing $x_1$ and $x_2$ (blue ellipses). When knowing $x_0, x_1, x_2$ (in blue), the matrix can be uniquely determined and the one-step ahead prediction coincides with the point $x_3=\bma x_2&x_1\ema \bma x_1&x_0\ema^{-1}x_2=  \bma-0.75&    0.4\ema^\top $.   
\end{example}
\medskip

We can extend the result to the case when the noise is present by considering PSD matrix completions. Indeed, given the bounds $\alpha\geq 0$ and $L>0$, we find that $M=\bma A&\bma v&V\ema \ema$ is such that $M^\top M\preceq \boundcomp$ if and only if $\hat{G}$ is a PSD completion of $\bound_\blockbound$ with $\Gamma=\bma L^2I_n&\alpha^2(t+1)I_t\ema$. 
Then, we can apply Proposition \ref{pr:no_noise}  with $Y=\dataxplus-B\datau$ and $\xinterp$ and $\Xinterp$ as in \eqref{eq:X_with_noise} repeatedly
 for each matrix completion $\boundcomp$. Then, we obtain	\be\label{eq:ell_un}
\mc \dataxplus(u)\supseteq	\hskip -0.9cm\bigcup_{\substack{\boundcomp \text{ compl. of } \bound_\blockbound\\ \Dcomp\succeq 0\,, \text{rank}(\Dcomp)=\text{rank}(\Xinterp)\,,}} \hskip -0.9cm\{ \xplus \mid (\xplus-{\blue Bu}-\ccomp)^\top \mc \Acomp (\yinterp-{\blue Bu}-\ccomp)\leq \Qcomp \}
\ee
where $
\Dcomp:= \Xinterp^\top \boundcomp \Xinterp -(\dataxplus-B\datau)^\top (\dataxplus-B\datau)\,,
$
and
$$
\begin{aligned}
	\mc \Acomp &=  I_{n}+(\dataxplus-B\datau)\Dcomp^{\dagger}(\dataxplus-B\datau)^{\top}\succ 0		\\
	\ccomp & =(\dataxplus-B\datau)(\boundcomp^{\frac{1}{2}}\Xinterp)^\dagger \boundcomp^{\frac{1}{2}}\xinterp\\
	\Qcomp &=\xinterp^{\top}\boundcomp^{\frac{1}{2}}(\Abound-\boundcomp^{\frac{1}{2}}\Xinterp \Dcomp^\dagger \Dcomp(\boundcomp^\frac{1}{2}\Xinterp)^\dagger)\boundcomp^{\frac{1}{2}}\xinterp\,.
\end{aligned}
$$
We write the inclusion ($\supseteq$) instead of equality, since we formally omit the boundary cases where $\text{rank}(D)=\text{rank}(\Xinterp)$; however, by continuity of $D$ with respect to the bound, we conjecture that the two sets in fact coincide. 
\begin{figure}
	\centering
	\includegraphics[width=0.22\textwidth]{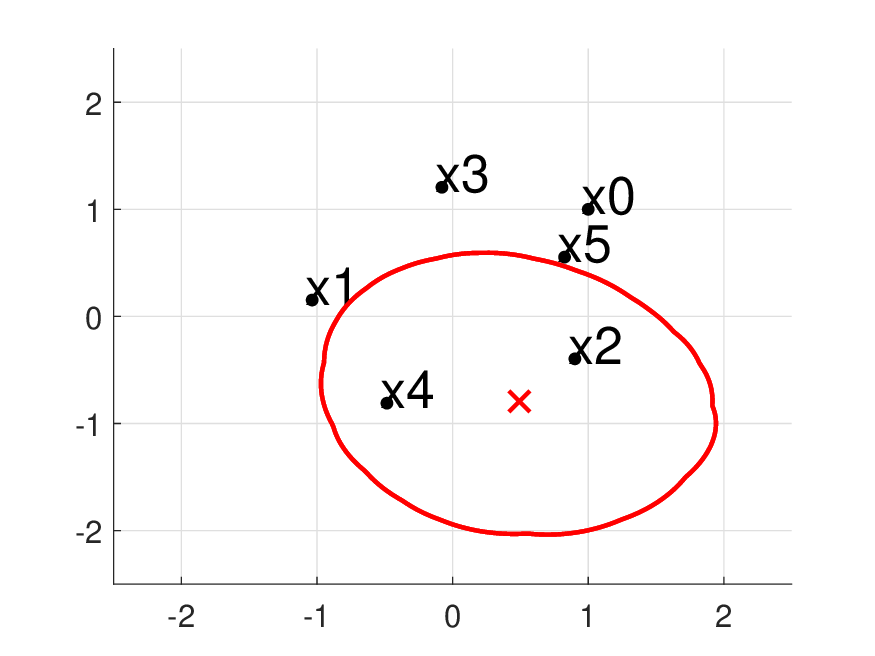}
    \includegraphics[width=0.22\textwidth]{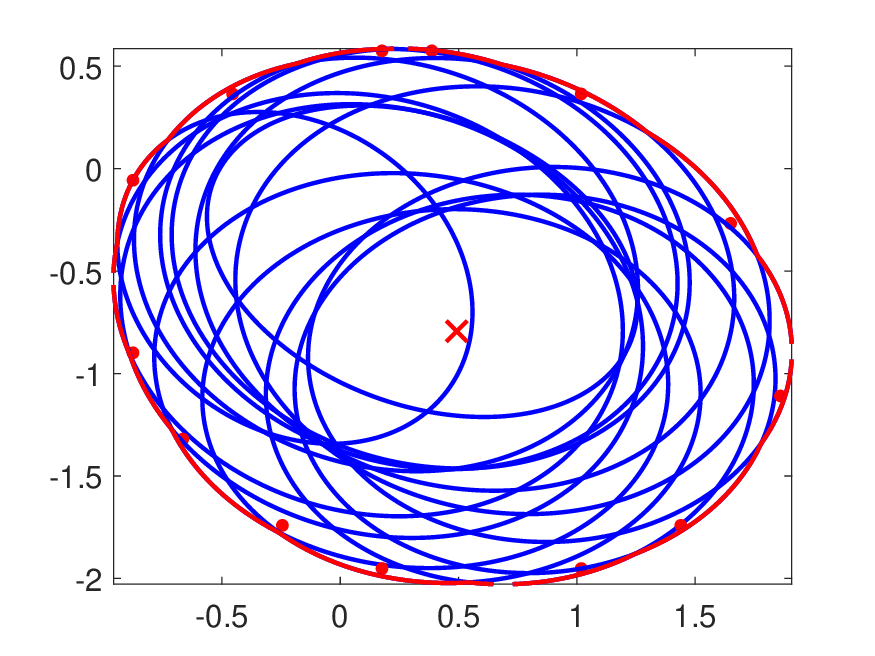}

\caption{Representation of Example \ref{ex:un_ell}}
\label{fig:un_ell}
\end{figure}
 \begin{example}\label{ex:un_ell}
Consider the data on the right of Fig. \ref{fig:datapoints}. The set of feasible states for $x_6$ for $L=1$ and $\alpha=0.32$ is plotted on the left of Fig. \ref{fig:un_ell} where the cross is the least-squares prediction. It is approximated as the union of several ellipsoids. Some of them are plotted on the right of Fig. \ref{fig:un_ell}. Observe that the set is convex and its shape suggests that it could be well over-approximated by an ellipsoid, which is the object of current work.
\end{example}
\medskip
The characterization in \eqref{eq:ell_un} offers a geometric interpretation of the feasible set at the subsequent state. Each feasible PSD completion corresponds to one admissible ellipsoid, and the union of these represents all future states consistent with the data and prior bounds. 
Moreover, safety and optimization problems can be formulated as subproblems over individual ellipsoids, enabling bilevel algorithms that alternate between searching for a matrix completion and optimizing over the corresponding ellipsoid. The exploration of these directions, as well as extensions to multi-step settings, constitutes the focus of ongoing and future research.
	\section{Conclusions}\label{sec:4}
	We derived interpolation conditions for noisy linear time-invariant systems with a known input matrix $B$, an unknown state-transition matrix $A$ with a bounded norm, and bounded process noise energy. These conditions enable a systematic characterization of all trajectories consistent with the measured data and prior bounds. 
We demonstrated their applicability for data consistency verification, inference of admissible system behaviors, and one-step-ahead prediction, providing a foundation for potential applications in safety verification and worst-case energy minimization. 
This work represents a first step toward a broader framework, with future directions including extensions to multi-step predictions, the development of algorithms for control applications, and the exploration of more general assumptions in both linear (e.g., unknown $B$, alternative prior knowledge on $A$, general noise bounds, stochastic setting) and nonlinear models.


		\bibliographystyle{ieeetr}
		\bibliography{bib}
		\appendix
		\subsection{Proof of Proposition \ref{pr:no_noise}}\label{app1}
		In the proof of Proposition \ref{pr:no_noise}, we will use twice the following Lemma.
		\begin{lemma}\label{lemma2}
			Let $D$ in \eqref{eq:D} 
			be such that $D\succeq 0$. Then, 	\be\label{eq:cond}(I_k-DD^\dagger)\Xinterp^{\top}=0\,.\ee  if and only if $\text{rank}(D)=\text{rank}(\Xinterp)$.
		\end{lemma}
		For the sake of completeness, we begin by proving the following Lemma.
		\begin{lemma}\label{lemma3}
			Let $v$ in $\R^{k\times m}$. Then, $(I_k-DD^\dagger)v=0$ if and only if   $\exists w\in \R^{k\times m}$ such that $v=Dw$ (i.e., $v \in \text{span}(D)$).
		\end{lemma}
		\begin{proof}[Lemma \ref{lemma3}]
			($\Rightarrow$) $(I_k-DD^\dagger)v=0$ if and only if $v=DD^\dagger v$. Then, $v=Dw$ for $w=D^\dagger v$. \\
			($\Leftarrow$) Let $v=Dw$. Then, $(I_k-DD^\dagger)Dw=0$ if and only if $Dw=DD^\dagger Dw=Dw$. 
		\end{proof}
	\medskip
		\begin{proof}[Lemma \ref{lemma2}]
			By Lemma \ref{lemma3}, \eqref{eq:cond} holds true if and only there exists $w$ in $\R^k$ such that
			\be\label{eq:gen_lin_sys}
			Dw=\Xinterp^{\top}\,.
			\ee
				By the Rouché-Capelli Theorem, the system \eqref{eq:gen_lin_sys} admits at least one solution if and only if
			\be\label{eq:rou-cap}
			\text{rank}(D)=	\text{rank}([D,\Xinterp^{\top}])\,.
			\ee
		Let $D$ be as in \eqref{eq:D} and such that $D\succeq 0$. 
        By Proposition \ref{pr:1}, $\exists \hat M$ 
			such that $\hat M^\top \hat M\preceq G$ and $Y=\hat M\Xinterp$.
		By substituting in \eqref{eq:D}, we find
			\be\label{eq:D2}
			D=\Xinterp^{\top}(G-\hat M^{\top}\hat M)\Xinterp\,.
			\ee
and			\be\label{eq:Db}
			\begin{aligned}
				\text{rank}(\big[D\,,\, \Xinterp^{\top}\big])&=\text{rank}(\Xinterp^{\top}\big[(G-\hat M^{\top}\hat M)\Xinterp,I_k\big])\\&= \text{rank}(\Xinterp)\,.
			\end{aligned}
			\ee
			Then, \eqref{eq:rou-cap} holds true if and only if $\text{rank}(\Xinterp)=\text{rank}(D)$. 
		\end{proof}
	\medskip
		\begin{proof}[Proposition \ref{pr:no_noise}]
			Let  $\bar Y=[y\,,\, Y]$ and $\bar \Xinterp=[\xinterp\,,\,\Xinterp]$.
			By Proposition \ref{pr:1},
			$\exists M$ 
			such that $M^\top M\preceq G$ and $\bar Y=M \bar \Xinterp$ 
			if and only if $\bar Y^{\top} \bar Y\preceq \bar \Xinterp^{\top} G \bar \Xinterp$, that is,  
			$$
				\left(\begin{matrix}
					y^{\top}y&y^{\top}Y\\
				\ast& Y^{\top} Y
				\end{matrix}\right)\preceq 
				\left(\begin{matrix}
					\xinterp^{\top}\,G\xinterp &\xinterp^{\top}\,G\Xinterp\\
				\ast& \Xinterp^{\top}\,G\Xinterp
				\end{matrix}\right)
			$$
			which is equivalent to
			\be \label{eq:1}
			\left(\begin{matrix}
				y^{\top}y- 	\xinterp^{\top}\,G \xinterp &y^{\top}Y-\xinterp^{\top}\,G\Xinterp\\
				\ast& Y^{\top} Y-\Xinterp^{\top} \,G\Xinterp
			\end{matrix}\right)\preceq0\,.
			\ee
			By the generalized Schur's complement, \eqref{eq:1} holds iff
			\be \label{eq:ell_schur}
			\begin{cases}
				D :=\Xinterp^{\top} G\Xinterp-Y^{\top} Y\succeq 0\\
				(I_k-DD^\dagger)B^{\top}=0\\
				y^{\top}y-\xinterp^{\top}G\xinterp+
				BD^{\dagger}B^{\top}\leq0
			\end{cases}
			\ee
			where $B=	y^{\top}Y-\xinterp^{\top}G\Xinterp$. 
			If $D\nsucceq 0$, the set of the feasible points at the next step is an empty set (leading to point $(1)$). 
	Let $D\succeq 0$. By Proposition \ref{pr:1}, there exists $\hat M$  such that  $\hat M^\top \hat M\preceq G$ and $Y=\hat M\Xinterp$.
			Then, for all $y$ in $\R^n$, 
			$$
			\begin{aligned}
				(I_k-DD^\dagger)&(Y^{\top} y- \Xinterp^{\top}G\, \xinterp)=\\&=(I_k-DD^\dagger)\Xinterp^{\top}(\hat M^{\top}y-G\xinterp)=0\,,
			\end{aligned}
			$$
			where the last equality holds true by Lemma \ref{lemma2} and the assumption that $\text{rank}(D)=\text{rank}(\Xinterp)$. 
			Then, 			the second equality is always satisfied and the set $\mc Y_+$ is entirely characterized by the last inequality. 	
		After some algebraic manipulations, we find that \eqref{eq:ell_schur} is equivalent to
			\be\label{eq:ell1_1}
			y^{\top}\mc A y-2\mc B^{\top}y+\mc C\leq  0
			\ee
			where
			\be\label{eq:ell_p0}
			\begin{aligned}
				\mc A &=  I_{n}+YD^{\dagger}Y^{\top};
				\\
				\mc B &= YD^{\dagger}\Xinterp^{\top}G \xinterp; \\
				\mc C &= \xinterp^{\top}G^\frac{1}{2}(G^\frac{1}{2}\Xinterp D^{\dagger}	\Xinterp^{\top}G^\frac{1}{2}-I_m)G^\frac{1}{2}\xinterp\,.
			\end{aligned}
			\ee
			Since $D\succeq 0$, we obtain $D^{\dagger}\succeq 0$ (see Corollary 3 in \cite{lewis1968pseudoinverses}) and therefore  $\mc A=I_n+YD^{\dagger}Y^\top\succ 0$ for every $Y$. Then, \eqref{eq:ell1_1} defines an ellipsoid and  $\mc A$ is invertible. Furthermore, since by definition $D=D^{\top}$ is symmetric, we have that $(D^\dagger)^{\top}=D^\dagger$, $\mc A^{{\top}}=\mc A$ and $(\mc A^{-1})^{\top}=\mc A^{-1}$ are also symmetric. This observation, and some algebraic manipulations, shows that \eqref{eq:ell1_1} can be rewritten in the form \eqref{eq:ell}
			with $\mc A$ as in \eqref{eq:ell_p0} and
			\be\label{ell:ps}
				c =\mc A^{-1}\mc B \,,\quad 
				Q =\mc B^{\top} \mc A^{-1}\mc B-\mc C\,.
			\ee
			{
			We further prove that the expression of $c$ and $Q$ in \eqref{ell:ps} can be equivalently rewritten as
					\be\label{eq:par_ell}
			\begin{aligned}
				c & =\Yinterp(\bound^{\frac{1}{2}}\Xinterp)^\dagger \bound^{\frac{1}{2}}\xinterp\\
				Q &=\xinterp^{\top}\bound^{\frac{1}{2}}(I_m-\bound^{\frac{1}{2}}\Xinterp D^\dagger D(\bound^\frac{1}{2}\Xinterp)^\dagger)\bound^{\frac{1}{2}}\xinterp\,.
			\end{aligned}
			\ee
			We start by deriving the expression for $c$, that is, 
				\be\label{eq:simpl}
        (I_{n}+YD^\dagger Y^{\top})^{-1}YD^\dagger\Xinterp^{\top}G\xinterp=Y(G^{\frac{1}{2}}\Xinterp)^\dagger G^{\frac{1}{2}}\xinterp
				\ee
				Note that \eqref{eq:simpl} is true for every $\xinterp$ if and only if
				$$
				\begin{array}{lcl}
					YD^\dagger \Xinterp^{\top}G&=&(I_{n}+YD^\dagger Y^{\top})Y(G^{\frac{1}{2}}\Xinterp)^\dagger G^{\frac{1}{2}}\\
			&=&Y(G^{\frac{1}{2}}\Xinterp)^\dagger G^{\frac{1}{2}} \\[2pt]&&+YD^\dagger(Y^{\top}Y-\Xinterp^{\top}G\Xinterp)(G^{\frac{1}{2}}\Xinterp)^\dagger G^{\frac{1}{2}} \\[2pt]&&+YD^\dagger \Xinterp^{\top}G\Xinterp(G^{\frac{1}{2}}\Xinterp)^\dagger G^{\frac{1}{2}}\\
					&=&Y(I_t-D^\dagger D)(G^{\frac{1}{2}}\Xinterp)^\dagger G^{\frac{1}{2}} \\[2pt]&&+YD^\dagger \Xinterp^{\top}G\Xinterp(G^{\frac{1}{2}}\Xinterp)^\dagger G^{\frac{1}{2}}\,.
				\end{array}
				$$
				Since $D\succeq 0$, $D$ is an EP matrix, that is, $D^\dagger D=DD^\dagger$ (see Theorem 2 in \cite{lewis1968pseudoinverses}). Then, 
				$$
				\begin{aligned}
				Y(I_k&-D^\dagger D)(G^{\frac{1}{2}}\Xinterp)^\dagger G^{\frac{1}{2}} \\&=Y(I_k-DD^\dagger)\Xinterp^{\top}(\Xinterp^\dagger)^{\top} (G^{\frac{1}{2}}\Xinterp)^\dagger G^{\frac{1}{2}}\overset{(\star)}{=}0  
				\end{aligned}
				$$
				where $(\star)$ follows by Lemma \ref{lemma2}.  Then, \eqref{eq:simpl} holds true iff
				$$
				YD^\dagger \Xinterp^{\top}G =YD^\dagger \Xinterp^{\top}G\Xinterp(G^{\frac{1}{2}}\Xinterp)^\dagger G^{\frac{1}{2}} 
				$$
				which is always true by the properties of the pseudo-inverse: 
                \be\label{eq:pseud}
				\begin{aligned}
					\Xinterp^{\top}G&=(G^\frac{1}{2}\Xinterp)^{\top}G^\frac{1}{2}=(G^\frac{1}{2}\Xinterp(G^\frac{1}{2}\Xinterp)^\dagger G^\frac{1}{2}\Xinterp)^{\top}G^\frac{1}{2}\\
					&=(G^\frac{1}{2}\Xinterp)^{\top}G^\frac{1}{2}\Xinterp(G^\frac{1}{2}\Xinterp)^\dagger G^\frac{1}{2}=\Xinterp^{\top}G\Xinterp(G^\frac{1}{2}\Xinterp)^\dagger G^\frac{1}{2}\,.
				\end{aligned}
				\ee
				Finally,  substituting the expression of $c=\mc A^{-1}\mc B$ in the definition of $Q$ in  \eqref{ell:ps} gives \eqref{eq:par_ell}. 
			This concludes the proof.}
		
	\end{proof}

	\end{document}